\documentclass[10pt, twocolumn, conference]{IEEEtran}    

\IEEEoverridecommandlockouts                              
\title{Diversity-Multiplexing Tradeoffs in\\
MIMO Relay Channels}

\author{%
  \authorblockN{Deniz G\"{u}nd\"{u}z\authorrefmark{1}\authorrefmark{2},
    Andrea Goldsmith\authorrefmark{2} and H. Vincent Poor\authorrefmark{1},
  }\\
  \authorblockA{%
    \authorrefmark{1}Department of Electrical Engineering, Princeton University, Princeton, NJ.\\
  }
  \authorblockA{%
    \authorrefmark{2}Department of Electrical Engineering, Stanford University, Stanford, CA.\\
  }
  Email: dgunduz@princeton.edu, andrea@wsl.stanford.edu, poor@princeton.edu
  \thanks{This research was supported by the National Science Foundation under Grants ANI-03-38807 and CNS-06-25637 and the DARPA ITMANET program under grant 1105741-1-TFIND.}
}
\date{November, 2007}
\usepackage{amsfonts}
\usepackage{amssymb}
\usepackage{amsmath}
\usepackage{amscd}
\usepackage{psfrag, graphicx}

\newtheorem{thm}{Theorem}[section]
\newtheorem{cor}[thm]{Corollary}

\newtheorem{prop}[thm]{Proposition}

\begin{document}
\maketitle
\thispagestyle{empty}
\pagestyle{empty}

\begin{abstract}
A multi-hop relay channel with multiple antenna terminals in a quasi-static slow fading environment is considered. For both full-duplex and half-duplex relays the fundamental diversity-multiplexing tradeoff (DMT) is analyzed. It is shown that, while decode-and-forward (DF) relaying achieves the optimal DMT in the full-duplex relay scenario, the dynamic decode-and-forward (DDF) protocol is needed to achieve the optimal DMT if the relay is constrained to half-duplex operation. For the latter case, static protocols are considered as well, and the corresponding achievable DMT performance is characterized.
\end{abstract}

\section{Introduction}

Relays are commonly used in wireless networks to improve performance, although the fundamental capacity limits of relay channels have yet to be fully characterized, even for simple systems \cite{Cover}. Rather than focus on capacity limits, we are interested in characterizing the tradeoff between rate gain through multiplexing versus the robustness gain through diversity associated with multiple-antenna relays. We will focus on a multiple antenna multi-hop system in which the source transmission can only be received by the relay terminal, as shown in Fig. \ref{f:system}. We call this the multiple-input multiple-output (MIMO) multi-hop relay channel. The links are assumed to be quasi-static, frequency non-selective Rayleigh fading, and the channel state information (CSI) is available only at the receiving end of each transmission.

We analyze this system in terms of the diversity-multiplexing tradeoff (DMT) in the high signal-to-noise ratio (SNR) regime introduced in \cite{Tse}. DMT analysis is useful in characterizing the fundamental tradeoff between the reliability and the number of degrees of freedom of a communication system. In DMT analysis, reliability is measured in terms of the diversity gain, which characterizes the rate of decay of the error probability with increasing SNR. The degrees of freedom is measured by the spatial multiplexing gain, which is the rate of increase in the transmission rate with SNR. While the DMT analysis is a tool to characterize the fundamental limits of a communication system in a fading environment, practical space-time codes that approach these theoretical limits have been designed \cite{HeshamCaire}- \cite{Murugan}.

In a channel with relays, the source's transmission is received by both the relays and the destination, and the source and the relay terminals cooperate to transmit the message to its intended destination \cite{Erkip, Laneman}. DMT analysis has been extensively applied to this general relay channel model; however, a full characterization of the DMT curve is still an open problem. In \cite{Azarian} the DMT of the half-duplex single-antenna relay channel is analyzed and a dynamic decode-and-forward (DDF) protocol is proposed. In DDF, the relay terminal listens to the source transmission until it can decode the message, and then starts transmitting the message jointly with the source terminal. The DMT of DDF is shown to dominate that of all other protocols, but for high multiplexing gains it does not meet the cut-set upper bound, which dictates the maximum possible of such gains \cite{Cover_book}. In \cite{Varanasi} DDF performance is improved slightly by using superposition coding. In \cite{Yuksel}, under the assumption of full CSI at the relay terminal, the compress-and-forward protocol is shown to achieve the optimal DMT performance. There has also been some recent interest in the DMT analysis for multi-hop relay systems; in \cite{Belfiore} and \cite{Hassibi} multiple single antenna relays operating in a distributed manner are considered. Due to the distributed nature of the relay nodes, amplify-and-forward relaying is considered, under which the achievable DMT is characterized.

\psfrag{M1}{$M_1$}
\psfrag{M2}{$M_2$}
\psfrag{M3}{$M_3$}
\psfrag{H1}{$\mathbf{H}_1$}
\psfrag{H2}{$\mathbf{H}_2$}
\psfrag{S}[][.3]{\huge{S}}
\psfrag{R}[][.3]{\huge{R}}
\psfrag{D}[][.3]{\huge{D}}
\begin{figure}
\centering
\includegraphics[width=3.5in]{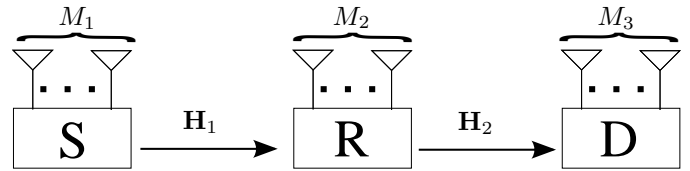}
\caption{The $(M_1, M_2, M_3)$ MIMO multi-hop relay channel model considered in the paper. There is no direct link from the source terminal (S) to the destination terminal (D).}
\label{f:system}
\end{figure}

In contrast to these prior works, we consider a MIMO multi-hop relay channel. For this model, the relay can decode the message without sacrificing degrees of freedom. While we only derive results for a single relay in this paper, our results can be extended to multiple relays. In the case of a full-duplex relay, we show that the DF protocol with a block Markov structure achieves the optimal DMT performance. In the half-duplex relay case, we first find the DMT of static protocols in which the source and the relay transmission periods are fixed, independent of the channel realization. On the other hand, we show that the DDF protocol of \cite{Azarian}, in which the time allocation depends on the realization of the source-relay channel, achieves the optimal DMT performance. In the multi-hop scenario, since the relay and the source do not transmit simultaneously, they do not need to use distributed space-time codes, which are harder to realize in practice \cite{JingHassibi, DamenHammons}. Furthermore, there is no need to inform the source or the destination terminals about the relay decision time as opposed to the general relay scenario. Hence, the dynamic relaying scheme in the case of the multi-hop relay channel can be realized by using an incremental redundancy code at the source \cite{Murugan}, and any DMT-optimal space-time code at the relay. Although the DMT of DDF has been previously shown to dominate that of other protocols in the case of general half-duplex relay channels, our results prove its optimality in the multi-hop relay scenario. In a concurrent work \cite{Gharan:IT:08}, Gharan et al. prove the optimality of the DDF protocol in a single-antenna multiple access relay network.


\section{System Model}\label{s:system}

We consider a three node multi-hop channel composed of source, relay and destination terminals with $M_1$, $M_2$ and $M_3$ antennas, respectively, as in Fig. \ref{f:system}. We call this system an $(M_1,M_2,M_3)$ multi-hop relay channel. The source-relay and the relay-destination channels are given by
\begin{eqnarray}
\mathbf{Y}_i &=& \sqrt{\frac{SNR}{M_i}} \mathbf{H}_i \mathbf{X}_i + \mathbf{W}_i,
\end{eqnarray}
for $i=1,2$, respectively, where $\mathbf{Y}_i$, $i=1,2$, are the received signals at the relay and the destination, respectively. Note that the source transmission is not received at the destination in our multi-hop relay channel model. Channels are assumed to be frequency non-selective, quasi-static Rayleigh fading and independent of each other; that is, for $i=1,2$, $\mathbf{H}_i$ is an $M_{i+1} \times M_i$ channel matrix whose entries are independent and identically distributed (i.i.d.) complex Gaussian random variables with zero means and unit variances (i.e., they are $\mathcal{CN}(0,1)$). The additive white Gaussian terms also have i.i.d. entries with $\mathcal{CN}(0,1)$. $\mathbf{X}_i$, $i=1,2$, are $M_i \times T$ source and relay input matrices, where $T$ is the total number of transmissions over which the channel is constant. We have short-term power constraints at the source and the relay given by $\mathrm{trace}(E[\mathbf{X}_i^H \mathbf{X}_i]) \leq M_i T$. For $i=1,2$, we define \[M_i^*\triangleq \min\{M_i, M_{i+1}\}.\] We assume that only the receivers have channel state information.

Following \cite{Tse}, for increasing $SNR$ we consider a family of codes and say that the system achieves a multiplexing gain of $r$ if the rate $R(SNR)$ satisfies \[\lim_{SNR \rightarrow \infty} \frac{R(SNR)}{\log(SNR)} = r.\] The diversity gain $d$ of this family is defined as \[d=-\lim_{SNR \rightarrow \infty} \frac{\log P_e(SNR)}{\log(SNR)},\] in which $P_e(SNR)$ is the error probability. For each $r$, define $d(r)$ as the supremum of the diversity gain over all families of codes. The full characterization of the DMT curve for a MIMO system is given in the following theorem \cite{Tse}.

\begin{thm}\label{t:Tse}
For a MIMO system with $M_1$ transmit and $M_2$ receive antennas and sufficiently long codewords, the optimal DMT curve $d_{M_1,M_2}(r)$ is given by the piecewise-linear function connecting the points $(k,d(k))$, $k=0,\ldots,\min(M_1,M_2)$, where $d(k)=(M_1-k)(M_2-k)$.
\end{thm}

For the rest of the paper, we always consider codes with sufficiently long codewords so that the error event is dominated by the outage event.

\section{DMT of MIMO Multi-hop Relay channels}\label{s:DMT}

\subsection{Full-duplex Relaying}\label{s:fdDMT}

We first consider the full-duplex relay case. The next theorem shows that the DMT tradeoff of the end-to-end system is equal to the worst-case DMT tradeoff of each link along the multi-hop path. The DMT characterization given here for a single relay can be easily generalized to multiple full-duplex relays.

\begin{thm}
The DMT $d^f_{M_1,M_2,M_3}(r)$ of an $(M_1,M_2,M_3)$ full-duplex system is characterized by
\begin{eqnarray}
  d^f_{M_1,M_2,M_3}(r) &=& \min \{d_{M_1,M_2}(r), d_{M_2,M_3}(r)\}.
\end{eqnarray}
\end{thm}
\begin{proof}
The result follows easily as DF achieves the capacity for a full-duplex multi-hop relay channel \cite{Cover}.
\end{proof}

\subsection{Static Protocols for Half-duplex Relaying}\label{s:hdsDMT}

In the half-duplex relay scenario, the total $T$ time units need to be divided among the source and the relay transmissions. We first consider static protocols for which the time allocation is fixed, independent of the channel states. However, similar to the generalized decode-and-forward protocol in \cite{Deniz_relay}, we consider unequal division of the time slot among the source and the relay. The source transmits during the first $aT$ channel uses, where $0<a<1$. The relay tries to decode the message and forwards over the remaining $(1-a)T$ channel uses. We call this protocol \emph{decode-and-forward with fixed allocation} (fDF), and its DMT is given in the next proposition.

\begin{prop}
The DMT of the half-duplex $(M_1,M_2,M_3)$ relay channel with fixed time allocation $a$ $(0<a<1)$ is
\begin{small}
\begin{eqnarray}\label{e:fDF}
d^{fDF}_{M_1M_2M_3}(r) &=& \min \left\{d_{M_1,M_2}\left(\frac{r}{a}\right), d_{M_2,M_3}\left(\frac{r}{1-a}\right)\right\}
\end{eqnarray}
\end{small}
\end{prop}
\begin{proof}
This result follows easily from Theorem \ref{t:Tse} with simple scaling of the DMT curve due to time division.
\end{proof}

\begin{figure}
\centering
\includegraphics[width=3.8in]{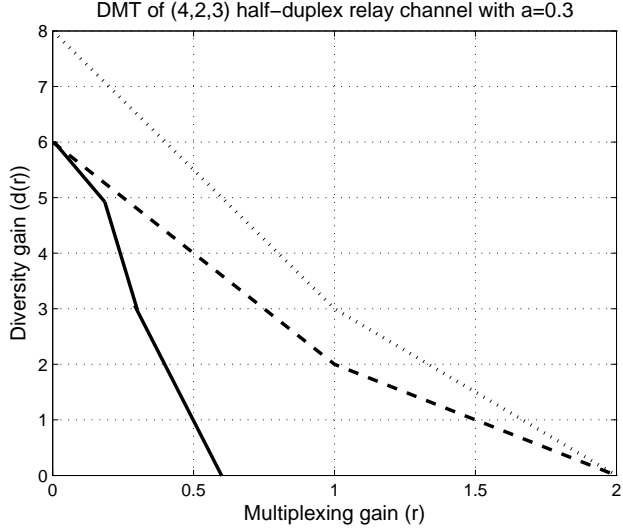}
\caption{The dotted and the dashed curves correspond to $d_{4,2}(r)$ and $d_{2,3}(r)$, respectively. Note that the dashed curve also corresponds to the DMT in the case of a full-duplex relay terminal. The solid curve is the DMT curve of a $(4,2,3)$ half-duplex multi-hop relay with the fDF protocol and $a=0.3$.}
\label{f:fDF}
\end{figure}

We can see from the above DMT that the highest multiplexing gain for the fDF scheme is $\min \{a M_1^*, (1-a)M_2^*\}$. On the other hand, the highest diversity gain is limited to $M_2\min\{M_1, M_3\}$. We illustrate the DMT of a $(4,2,3)$ system with a fixed time allocation of $a=0.3$ in Fig. \ref{f:fDF}.

Since different time allocations result in different DMT curves, we can optimize the time allocation based on the multiplexing gain \cite{Deniz_relay, EliaKumar}. We call this protocol \emph{DF with variable time allocation} (vDF). Note that this is still a static protocol since the time allocation variable is determined based only on the multiplexing gain and is independent of the channel realization. For each multiplexing gain $r$, the diversity gain is the minimum of the two diversity gains in (\ref{e:fDF}); hence the optimal time allocation variable $a(r)$ is the one that satisfies
\begin{eqnarray}\label{e:vDF}
d^{vDF}_{M_1,M_2,M_3}(r) = d_{M_1,M_2}\left(\frac{r}{a(r)}\right) = d_{M_2,M_3}\left(\frac{r}{1-a(r)}\right).
\end{eqnarray}

\begin{cor}
The number of degrees of freedom of an $(M_1,M_2,M_3)$ multi-hop relay channel with the vDF protocol is $\frac{M_1^*M_2^*}{M_1^*+M_2^*}$, while the maximal diversity gain is $M_2\min\{M_1, M_3\}$.
\end{cor}

We now present the DMT for some special cases because a general closed form expression is not tractable. We first consider the $(M_1,1,M_3)$ system. Since the hops for this setup are multiple-input single-output (MISO) and single-input multiple-output (SIMO) systems, the DMTs are characterized as $d_{M_i,M_{i+1}}=M_i^*(1-r)$, $i=1,2$. From (\ref{e:vDF}) and defining $A\triangleq M_1^*/M_3^*$ and $B \triangleq 1-r-A(1+r)$, we find \[a(r) = \frac{-B+\sqrt{B^2-4A(A-1)r}}{2(A-1)}\]  for $A\neq 1$. We have $a(r)=0.5$ if $A=1$. The DMT achieved by the vDF protocol in a $(4,1,3)$ system is plotted in Fig. \ref{f:vDF1}. In this figure, we also plot the DMT for the fDF scheme with a fixed time allocation $a=0.5$.

If we have $M_1=M_3=M$, then the optimal time allocation is $a=0.5$ independent of the multiplexing gain, and the DMT is given by $d^{vDF}_{M,M_2,M}(r) = d_{M,M_2}(2r)$.

\begin{figure}
\centering
\includegraphics[width=3.8in]{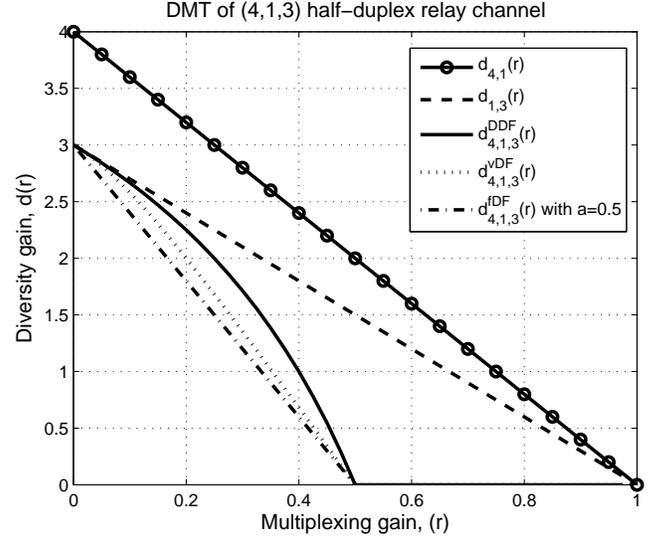}
\caption{The DMT curve of a $(4,1,3)$ multi-hop relay channel. The two topmost curves correspond to the cut-set bounds, where the dashed curve is also the DMT for a full-duplex relay. The DDF, vDF and fDF protocol with $a=0.5$ are also illustrated, where the DDF curve is the optimal DMT with half-duplex relaying.}
\label{f:vDF1}
\end{figure}

\subsection{Dynamic Decode-and-Forward Protocol for Half-duplex Relaying}\label{s:hddDMT}

In \cite{Azarian}, Azarian et al. proposed the dynamic decode-and-forward protocol for the cooperative relay channel with single antennas. In DDF for the relay channel, the source transmits during the entire timeslot using an incremental redundancy type codebook. This code design enables the relay to decode the message after receiving only a portion of the codeword; hence the relay decodes the message when the accumulated mutual information over the source-relay channel is sufficient for the transmission rate. Thus, the relay decoding time becomes a random variable that depends on the source-relay channel quality. As soon as the relay decodes the message, it starts transmitting.

The achievable DMT of the DDF scheme in the case of a single-antenna cooperative relay channel is characterized in \cite{Azarian}, where it is shown to dominate the DMTs of amplify-and-forward (AF) and decode-and-forward (DF) based protocols and, more strikingly, to achieve the DMT upper bound for multiplexing gains $r<0.5$. Hence, DDF is DMT-optimal in this range of low multiplexing gains for the single antenna cooperative relay channel.

Here, we consider using the DDF protocol for the multi-antenna multi-hop relay channel, and show that it achieves the DMT cut-set upper bound; that is, DDF is DMT-optimal. The intuitive explanation behind the optimality of DDF in this setting is as follows: In the multi-hop relay scenario, the message needs to be decoded at the relay terminal, since otherwise the destination would not be able to decode it either, due to the data processing inequality. However, any fixed time allocation scheme either wastes multiplexing gain, since it cannot utilize the good states of the source-relay channel, or results in outage in the case of a poor quality (low $\mathrm{SNR}$) source-relay channel. DDF, by enforcing decoding at the relay and dynamically allocating the source transmission time based on the source-relay channel state, achieves the optimal DMT performance.

\begin{thm}\label{t:DDF_ach}
For the $(M_1,M_2,M_3)$ system with rate $R=r\log SNR$, the outage probability of the DDF is given by
\[P_{out}(r) \doteq SNR^{-d^{DDF}(r)}\]
where
\begin{equation}\label{DDF_opt}
d^{DDF}(r) = \inf_{(\boldsymbol{\alpha_1}, \boldsymbol{\alpha_2}) \in \mathcal{\tilde{O}}_2} \sum_{i=1}^2 \sum_{j=1}^{M_i^*}(2j-1+|M_i-M_{i+1}|)\alpha_{i,j}
\end{equation}
and
\begin{eqnarray}
\mathcal{\tilde{O}}_2 &\triangleq& \left\{ (\boldsymbol{\alpha_1}, \boldsymbol{\alpha_2}) \in \mathcal{R}^{M_1^*+} \times \mathcal{R}^{M_2^*+} \bigg| \nonumber \right. \\
&& \left. \alpha_{i,1}\geq\cdots\geq\alpha_{i,M_i^*}\geq0 , r > \frac{S_1(\boldsymbol{\alpha_1}) S_2(\boldsymbol{\alpha_2})} {S_1(\boldsymbol{\alpha_1})+S_2(\boldsymbol{\alpha_2})}  \right\}  \nonumber
\end{eqnarray}
in which we have defined
\begin{equation}
S_i(\boldsymbol{\alpha_i}) \triangleq \sum_{j=1}^{M_i^*} (1-\alpha_{i,j})^+, \mbox{ for $i=1,2$}.
\end{equation}
\end{thm}
\begin{proof}
The proof can be found in Appendix \ref{a:DDF_ach}.
\end{proof}

The optimality of the above DMT achieved by the DDF protocol is shown in the following theorem.

\begin{thm}\label{t:DDF_ub}
DDF is DMT-optimal for MIMO multi-hop half-duplex relay channels.
\end{thm}
\begin{proof}
The proof can be found in Appendix \ref{a:DDF_ub}.
\end{proof}

\begin{cor}\label{c:div_mux_DDF}
The number of degrees of freedom of an $(M_1,M_2,M_3)$ multi-hop relay channel is $\frac{M_1^*M_2^*}{M_1^*+M_2^*}$, while the maximal diversity gain is $M_2\min\{M_1, M_3\}$. Hence, the end-points of the DMT curve can also be achieved by static relaying, i.e., with fixed time allocation corresponding to the multiplexing gain.
\end{cor}

It can be seen from Theorem \ref{t:DDF_ach} that the DMT of a half-duplex multi-hop relay channel is not a piecewise-linear function as in the case of a point-to-point MIMO channel. While it is hard to give a general closed form expression for the DMT of MIMO multi-hop channels, for given $M_1, M_2$ and $M_3$ and a fixed multiplexing gain $r$, the optimization problem in (\ref{DDF_opt}) can be converted into a convex optimization problem, and hence can be solved efficiently \cite{Boyd}. We now give an explicit characterization of the DMT for some classes of multi-hop relay channels.

\begin{figure}
\centering
\includegraphics[width=3.8in]{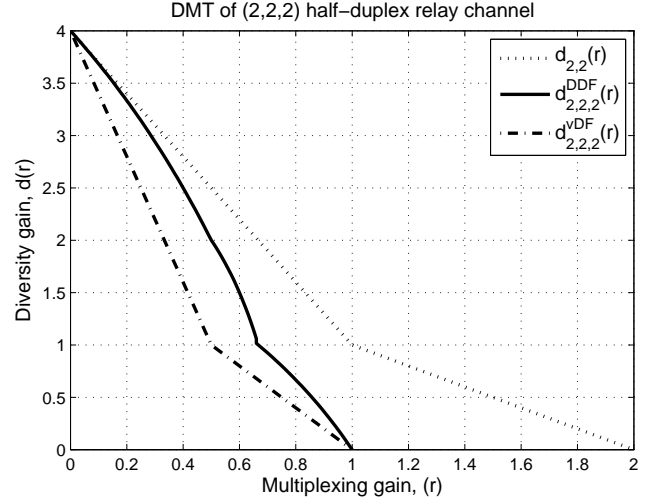}
\caption{The DMT of a $(2,2,2)$ system. From top to bottom, the three curves correspond to the full-duplex relay DMT, the half-duplex relay DMT which is achievable by DDF protocol, and the DMT of the static protocol with $a=0.5$.}
\label{f:DDF2}
\end{figure}

\begin{cor}
The DMT of an $(M_1,1,M_3)$ system is \[d^{DDF}_{M_1,1,M_3}(r) = \min(M_1,M_3)\frac{1-2r}{1-r} \] for $0 \leq r \leq 1/2$, and $0$ elsewhere.
\end{cor}

In Fig. \ref{f:vDF1} we illustrate the DMT of the $(4,1,3)$ multi-hop relay channel, which is achieved by the DDF protocol. We see that the DDF DMT dominates that of the static protocols at all multiplexing gains except the end-points. As stated in Corollary \ref{c:div_mux_DDF} these end-points can be achieved by the static fixed time allocation fDF protocol as well.

\begin{cor}
The DMT of the $(2,2,2)$ system with a half-duplex relay is given by
\begin{eqnarray}
d^{DDF}(r) = \left\{
\begin{array}{lll}
                \frac{2(4-5r)}{2-r} & \mbox{if} & 0 \leq \textit{b} < 1/2 \vspace{.02in} \\
                \frac{3-4r}{1-r} & \mbox{if} & 1/2 \leq b < 2/3 \vspace{.02in} \\
                \frac{4(1-r)}{2-r} & \mbox{if} & 2/3 \leq b \leq 1
           \end{array} \right.
\end{eqnarray}
\end{cor}

The DMT of the $(2,2,2)$ system is plotted in Fig. \ref{f:DDF2}. The topmost curve in the figure is the DMT of a $2 \times 2$ MIMO system, which can be achieved by a full-duplex relay. The lowest curve is the DMT of the vDF protocol. Note that for this symmetric scenario vDF reduces to fDF with $a=0.5$.

\section{Conclusions}\label{s:conc}

We have derived the diversity-multiplexing tradeoff of MIMO multi-hop relay channels for both full-duplex and half-duplex relays. For full-duplex relays, it is easy to show that the decode-and-forward protocol achieves the optimal DMT, which is simply the minimum of the DMTs of the two links. This applies to multiple relays as well; that is, the DMT of the end-to-end system will be limited by link with the smallest DMT. In the case of a half-duplex relay, we have shown that the dynamic decode-and-forward protocol, in which the relay listens until decoding and then forwards, achieves the optimal DMT, which is no longer a piecewise-linear function of the multiplexing gain. We have also shown that this optimal DMT performance cannot be achieved by static time allocation. Finally, we have provided explicit expressions for the DMT of some classes of half-duplex multi-hop relay systems, and compared the achievable performances with fixed and dynamic time allocation.

\appendices
\section{Proof of Theorem \ref{t:DDF_ach}}\label{a:DDF_ach}

For the achievability scheme, we assume that the inputs at both the source and the relay are Gaussian with identity covariance matrices. Let the transmission rate be $R=r\log \mathrm{SNR}$, and define
\begin{equation}
 C_i(\mathbf{H}_i) \triangleq \log \det \left(\mathbf{I} + \frac{\mathrm{SNR}}{M_i} \mathbf{H}_i\mathbf{H}_i^\dag \right).
\end{equation}
The relay listens for $aT$ channel uses until it decodes the message. Hence, we have \[a = \frac{r \log \mathrm{SNR}}{C_1}.\]

If $a \geq 1$ then the relay is in outage, which leads to an outage for the whole system. If $a<1$, then the relay transmits during the rest of the timeslot for $(1-a)T$ channel uses. Conditioned on successful decoding at the relay, i.e., $a<1$, the outage probability over the second hop is given by
\begin{align}
&P\{r\log \mathrm{SNR} > (1-a) C_2(\mathbf{H}_2)\} \nonumber\\
& = P\left\{r\log \mathrm{SNR} > \left(1-\frac{r \log \mathrm{SNR}}{C_1(\mathbf{H}_1)}\right)  C_2(\mathbf{H}_2)\right\} \nonumber \\
& = P\left\{r\log \mathrm{SNR} > \frac{C_1(\mathbf{H}_1) C_2(\mathbf{H}_2)}{C_1(\mathbf{H}_1)+C_2(\mathbf{H}_2)}. \label{e:DDFsb}\right\}
\end{align}

Let $\lambda_{i,1}, \ldots, \lambda_{1,M_i^*}$ be the nonzero eigenvalues of $\mathbf{H}_i\mathbf{H}_i^\dag$ for $i=1,2$. Suppose $\lambda_{i,j} =\mathrm{ SNR}^{-\alpha_{i,j}}$ for $j=1,\ldots, M_i^*$, $i=1,2$. We have\footnote{Define the exponential equality as $f(SNR) \dot{=} SNR^c$, if $\lim_{SNR \rightarrow \infty} \frac{\log f(SNR)}{\log SNR} = c$. The exponential inequalities $\dot{\leq}$ and $\dot{\geq}$ are defined similarly.}

\begin{eqnarray}\label{e:DDFout1}
C_i(\mathbf{H}_i) &= & \log \prod_{j=1}^{M_i^*} \left(1+\frac{\mathrm{SNR}}{M_i} \lambda_{i,j} \right) \nonumber \\
 &\doteq & \log \prod_{j=1}^{M_i^*} \mathrm{SNR}^{(1-\alpha_{i,j})^+} \label{e:DDFout2}
\end{eqnarray}
where $(x)^+ \triangleq \max\{0,x\}$. Using these exponential equalities, we can rewrite (\ref{e:DDFsb}) as follows
\begin{align}
& P\{r\log \mathrm{SNR} > (1-a) C_2(\mathbf{H}_2) \} \nonumber\\
&= P\left\{r\log \mathrm{SNR} > \frac{C_1(\mathbf{H}_1) C_2(\mathbf{H}_2)}{C_1(\mathbf{H}_1)+C_2(\mathbf{H}_2)} \right\} \nonumber \\
& \doteq P\left\{\log \mathrm{SNR}^r > \frac{\log \mathrm{SNR}^{S_1(\boldsymbol{\alpha_1})} \log \mathrm{SNR}^{S_2(\boldsymbol{\alpha_2})}}{\log \mathrm{SNR}^{S_1(\boldsymbol{\alpha_1})} + \log \mathrm{SNR}^{S_2(\boldsymbol{\alpha_2})}} \right\} \nonumber \\
&= P\left\{r > \frac{S_1(\boldsymbol{\alpha_1}) S_2(\boldsymbol{\alpha_2})} {S_1(\boldsymbol{\alpha_1}) + S_2(\boldsymbol{\alpha_2})} \right\} \nonumber
\end{align}
where we have $S_i(\boldsymbol{\alpha_i}) = \sum_{j=1}^{M_i^*} (1-\alpha_{i,j})^+$.

Then the overall outage probability can be written as
\begin{eqnarray}
P_{out}(r) &\doteq& P\left\{ r \geq S_1(\boldsymbol{\alpha_1}) \right\} \nonumber\\
&& + P\left\{ S_1(\boldsymbol{\alpha_1}) > r > \frac{S_1(\boldsymbol{\alpha_1}) S_2(\boldsymbol{\alpha_2})} {S_1(\boldsymbol{\alpha_1})+S_2(\boldsymbol{\alpha_2})}  \right\} \nonumber
\end{eqnarray}

We define
\begin{eqnarray}
\mathcal{O}_1 &\triangleq& \{ (\boldsymbol{\alpha_1}, \boldsymbol{\alpha_2}) : r \geq S_1(\boldsymbol{\alpha_1}) \} \nonumber \\
\mathcal{O}_2 &\triangleq& \left\{ (\boldsymbol{\alpha_1}, \boldsymbol{\alpha_2}) :  S_1(\boldsymbol{\alpha_1}) > r > \frac{S_1(\boldsymbol{\alpha_1}) S_2(\boldsymbol{\alpha_2})} {S_1(\boldsymbol{\alpha_1})+S_2(\boldsymbol{\alpha_2})}  \right\} \nonumber \\
\mathcal{O} &\triangleq& \mathcal{O}_1 \cup \mathcal{O}_2 \nonumber
\end{eqnarray}

Then using the joint probability of the eigenvalues of $\mathbf{H}_i\mathbf{H}_i^\dag$ given in \cite{Tse}, the outage probability can be computed as
\begin{eqnarray}
P_{out}(r) &\doteq& \int_{\mathcal{O}} p(\boldsymbol{\alpha_1}, \boldsymbol{\alpha_2}) d\boldsymbol{\alpha_1} d\boldsymbol{\alpha_2} \nonumber \\
&\doteq& \int_{\mathcal{O}'} \prod_{i=1}^2 \prod_{j=1}^{M_i^*} SNR^{-(2j-1+|M_i-M_{i+1}|)\alpha_{i,j}} d\boldsymbol{\alpha_1} d\boldsymbol{\alpha_2} \nonumber
\end{eqnarray}
where $\mathcal{O}' \triangleq \mathcal{O} \cap (\mathcal{R}^{M_1^*+}, \mathcal{R}^{M_2^*+})$.

Using Laplace's method as in \cite{Tse}, we obtain the exponential behavior of the outage probability as $P_{out}(r) \doteq SNR^{-d^{DDF}(r)}$, where
\begin{equation}
d^{DDF}(r) = \inf_{(\boldsymbol{\alpha_1}, \boldsymbol{\alpha_2}) \in\mathcal{O}'} f(\boldsymbol{\alpha_1}, \boldsymbol{\alpha_2})
\end{equation}
and
\begin{equation}
f(\boldsymbol{\alpha_1}, \boldsymbol{\alpha_2}) \triangleq \sum_{i=1}^2 \sum_{j=1}^{M_i^*}(2j-1+|M_i-M_{i+1}|)\alpha_{i,j}.
\end{equation}

Next, we define
\begin{eqnarray}
\mathcal{\tilde{O}}_1 &\triangleq& \left\{(\boldsymbol{\alpha_1}, \boldsymbol{\alpha_2}) \in \mathcal{R}^{M_1^*+} \times \mathcal{R}^{M_2^*+}| \nonumber \right. \\
&& \left. \alpha_{i,1}\geq\cdots\geq\alpha_{i,M_i^*}\geq0 , r \geq S_1(\boldsymbol{\alpha_1}) \right\} \nonumber \\
\mathcal{\tilde{O}}_2 &\triangleq& \left\{ (\boldsymbol{\alpha_1}, \boldsymbol{\alpha_2}) \in \mathcal{R}^{M_1^*+} \times \mathcal{R}^{M_2^*+}|  \nonumber \right. \\
&& \left. \alpha_{i,1}\geq\cdots\geq\alpha_{i,M_i^*}\geq0 , r > \frac{S_1(\boldsymbol{\alpha_1}) S_2(\boldsymbol{\alpha_2})} {S_1(\boldsymbol{\alpha_1})+S_2(\boldsymbol{\alpha_2})}  \right\} \nonumber \\
\mathcal{\tilde{O}} &\triangleq& \mathcal{\tilde{O}}_1 \cup \mathcal{\tilde{O}}_2 \nonumber
\end{eqnarray}
We can see that $\mathcal{O}' = \mathcal{\tilde{O}}$. Hence,
\begin{align}
&d^{DDF}(r) = \inf_{(\boldsymbol{\alpha_1}, \boldsymbol{\alpha_2}) \in \mathcal{\tilde{O}}} f(\boldsymbol{\alpha_1}, \boldsymbol{\alpha_2}) \nonumber \\
&= \inf \left\{\inf_{(\boldsymbol{\alpha_1}, \boldsymbol{\alpha_2}) \in \mathcal{\tilde{O}}_1} f(\boldsymbol{\alpha_1}, \boldsymbol{\alpha_2}), \inf_{(\boldsymbol{\alpha_1}, \boldsymbol{\alpha_2}) \in \mathcal{\tilde{O}}_2} f(\boldsymbol{\alpha_1}, \boldsymbol{\alpha_2}) \right\} \nonumber \\
&= \inf_{(\boldsymbol{\alpha_1}, \boldsymbol{\alpha_2}) \in \mathcal{\tilde{O}}_2} f(\boldsymbol{\alpha_1}, \boldsymbol{\alpha_2}) \nonumber
\end{align}
in which the last equality follows since we have
\[S_1(\boldsymbol{\alpha_1}) \geq \frac{S_1(\boldsymbol{\alpha_1}) S_2(\boldsymbol{\alpha_2})} {S_1(\boldsymbol{\alpha_1})+S_2(\boldsymbol{\alpha_2})}\]
for all $(\boldsymbol{\alpha_1}, \boldsymbol{\alpha_2})$, and hence $\mathcal{\tilde{O}}_1 \subseteq \mathcal{\tilde{O}}_2$.

\section{Proof of Theorem \ref{t:DDF_ub}}\label{a:DDF_ub}

We first give an upper bound for the DMT of the MIMO multi-hop half duplex relay channel, and show that the DDF DMT given in Theorem \ref{t:DDF_ach} matches this upper bound. Let $a \in (0,1]$ be the portion of the source transmit time, i.e., the source transmits over the first $aT$ channel uses. Hence, the relay transmits over the remaining $(1-a) T$ channel uses. Here we assume that the time allocation is independent of the message, i.e., it cannot be used for information transmission. As shown in \cite{Yuksel} this does not affect the DMT of the system.

From the two cut-set bounds, the instantaneous capacity $C(\mathbf{H}_1, \mathbf{H}_2)$ is upper bounded by \[ \max_{a, P_{\mathbf{X}_1} P_{\mathbf{X}_2}} \min \{a I(\mathbf{X}_1;\mathbf{Y}_1|\mathbf{H}_1), (1-a) I(\mathbf{X}_2;\mathbf{Y}_2|\mathbf{H}_2)\}.\]
Since the capacity is maximized with Gaussian inputs, and linear scaling of the power constraint does not affect the high SNR analysis, the instantaneous capacity can be bounded as \[C(\mathbf{H}_1, \mathbf{H}_2) \leq \max_{a} \min \{a \bar{C}_1(\mathbf{H}_1), (1-a) \bar{C}_2(\mathbf{H}_2) \},\] where we define \[\bar{C}_i(\mathbf{H}_i) \triangleq \log \det (\mathbf{I} + \mathrm{SNR} \mathbf{H}_i\mathbf{H}_i^\dag).\] We can further upper bound the capacity by assuming optimal time allocation at each channel realization. The instantaneous capacity is maximized at each channel realization for \[a(\mathbf{H}_1, \mathbf{H}_2) = \frac{\bar{C}_2(\mathbf{H}_2)}{\bar{C}_1(\mathbf{H}_1)+\bar{C}_2(\mathbf{H}_2)},\] and the corresponding upper bound is \[C(\mathbf{H}_1, \mathbf{H}_2) \leq \frac{\bar{C}_1(\mathbf{H}_1) \bar{C}_2(\mathbf{H}_2)} {\bar{C}_1(\mathbf{H}_1)+\bar{C}_2(\mathbf{H}_2)}.\]

For a transmission rate of $R=r\log SNR$, the outage probability lower bound is given by
\begin{eqnarray}
P_{out}(r) \geq P\left\{r\log \mathrm{SNR} > \frac{\bar{C}_1(\mathbf{H}_1) \bar{C}_2(\mathbf{H}_2)} {\bar{C}_1(\mathbf{H}_1)+\bar{C}_2(\mathbf{H}_2)} \right\}. \label{e:ubDDF1}
\end{eqnarray}


Using the characterization of the eigenvalues of the channel matrices given in Appendix \ref{a:DDF_ach}, we obtain
\begin{eqnarray}
P_{out}(r) & \dot{\geq} & P\left\{r > \frac{S_1(\boldsymbol{\alpha}_1) S_2(\boldsymbol{\alpha}_2)}{S_1(\boldsymbol{\alpha}_1) + S_2(\boldsymbol{\alpha}_2)} \right\} \nonumber
\end{eqnarray}
in which $S_i(\boldsymbol{\alpha}_i)$ is as defined before. Next, we define \[\mathcal{\bar{O}} \triangleq \left\{ (\boldsymbol{\alpha}, \boldsymbol{\beta}) : r > \frac{S_1(\boldsymbol{\alpha}_1) S_2(\boldsymbol{\alpha}_2)}{S_1(\boldsymbol{\alpha}_1) + S_2(\boldsymbol{\alpha}_2)} \right\}. \] Then the outage probability is lower bounded by
\begin{eqnarray}
P_{out}(r) & \dot{\geq}& \int_{\mathcal{\bar{O}}} p(\boldsymbol{\alpha}, \boldsymbol{\beta}) d\boldsymbol{\alpha} d\boldsymbol{\beta} \nonumber \\
&\doteq& \int_{\mathcal{\bar{O}}'} \prod_{i=1}^2 \prod_{j=1}^{M_i^*} SNR^{-(2j-1+|M_i-M_{i+1}|)\alpha_{i,j}} d\boldsymbol{\alpha_1} d\boldsymbol{\alpha_2} \nonumber
\end{eqnarray}
where \[\mathcal{\bar{O}}' \triangleq \mathcal{\bar{O}} \cap (\mathcal{R}^{M_1^*+}, \mathcal{R}^{M_2^*+}).\]

Similar to Appendix \ref{a:DDF_ach}, we obtain the exponential behavior of the above outage probability using Laplace's method. Note that, since $\mathcal{\bar{O}}' = \mathcal{\tilde{O}}_2$, the outage probability upper bound has the same diversity gain function as the DDF protocol. Hence, DDF is DMT-optimal.


\begin{thebibliography}{1}
\bibitem {Cover} T. M. Cover and A. A. El Gamal, ``Capacity theorems for the relay channel,'' \textit{IEEE Trans. Inf. Theory}, vol. 25, no. 5, pp. 572-584, Sep. 1979.
\bibitem {Tse} L. Zheng and D. Tse, ``Diversity and multiplexing: A fundamental tradeoff in multiple antenna channels,'' \textit{IEEE Trans. Inf. Theory}, vol. 49, no. 5, pp. 1073-1096, May 2003.
\bibitem {HeshamCaire} H. El Gamal, G. Caire and M. O. Damen, ``Lattice coding and decoding achieve the optimal diversity-vs-multiplexing tradeoff of MIMO channels,'' \textit{IEEE Trans. Inf. Theory}, vol. 50, no. 6, pp. 968-985, June 2004.
\bibitem {EliaKumarPawar} P. Elia, K. R. Kumar, S. A. Pawar, P. V. Kumar and H. Lu, ``Explicit, minimum-delay space-time codes achieving the diversity multiplexing gain tradeoff,''  \textit{IEEE Trans. Inf. Theory},  vol. 52, no. 9, pp. 3869-3884, Sep. 2006.
\bibitem {Oggier} F. Oggier, G. Rekaya, J-C. Belfiore and E. Viterbo, ``Perfect space-time block codes,'' \textit{IEEE Trans. Inf. Theory}, vol. 52, no. 9, pp. 3885-3902, Sep. 2006.
\bibitem {Murugan} A. Murugan, K. Azarian and H. El Gamal, ``Cooperative lattice coding and decoding in half-duplex channels,'' \textit{IEEE Jour. On Select Areas in Commun.}, vol. 25, no. 2, pp. 268-279, Feb. 2007.
\bibitem {Erkip} A. Sendonaris, E. Erkip and B. Aazhang, ``User cooperation diversity, Part I: System description,'' \textit{IEEE Trans. on Commun.}, vol. 51, no. 11, pp. 1927-1938, Nov. 2003.
\bibitem {Laneman} J. N. Laneman, D. N. C. Tse and G. W. Wornell, ``Cooperative diversity in wireless networks: Efficient protocols and outage behavior,'' \textit{IEEE Trans. Inf. Theory}, vol. 50, no. 12, pp. 3062-3080, Dec. 2004.
\bibitem  {Azarian} K. Azarian, H. El Gamal and P. Schniter, ``On the achievable diversity-multiplexing tradeoff in half-duplex cooperative channels,''  \textit{IEEE Trans. Inf. Theory}, vol. 51, no. 12, pp. 4152-4172, Dec. 2005.
\bibitem  {Cover_book} T. M. Cover and J. A. Thomas, \textit{Elements of Information Theory}, New York: Wiley-Interscience, 1991.
\bibitem {Varanasi} N. Prasad and M. Varanasi, "High performance static and dynamic cooperative communication protocols for the half duplex fading relay channel," \textit{Proc. IEEE Global Commun. Conf.}, San Fransisco, CA, Nov. 2006.
\bibitem {Yuksel} M. Yuksel and E. Erkip, ``Multi-antenna cooperative wireless systems: A diversity multiplexing tradeoff perspective.'' \textit{IEEE Trans. Inf. Theory}, vol. 53, no. 10, pp. 3371-3393, Oct. 2007.
\bibitem {Belfiore} S. Yang and J. C. Belfiore, ``Diversity of MIMO multihop relay channels,'' submitted to \textit{IEEE Trans. Inf. Theory}.
\bibitem {Hassibi} C. Rao and B. Hassibi, ``Diversity-multiplexing gain trade-off of a MIMO system with relays,'' \textit{Proc. IEEE Inf. Theory Workshop}, Bergen, Norway, July 2007.
\bibitem {JingHassibi} Y. Jing and B. Hassibi, ``Distributed space-time coding in wireless relay networks,'' \textit{IEEE Trans. Wireless Commun.}, vol. 5, no. 12, pp. 3524-3536, Dec. 2006.
\bibitem {DamenHammons} M. O. Damen and R. Hammons, ``Distributed space-time codes: relays delays and code word overlays,'' \textit{Proc. Int'l Conf. on Wireless Comm. and Mobile Computing}, Honolulu, HI, Aug. 2007.
\bibitem  {Gharan:IT:08} S. O. Gharan, A. Bayesteh and A. K. Khandani, ``On the diversity-multiplexing tradeoff in multiple-relay network,'' submitted to \textit{IEEE Trans. Inf. Theory}.
\bibitem  {Deniz_relay} D. G\"{u}nd\"{u}z and E. Erkip, ``Source and channel coding for cooperative relaying,'' \textit{IEEE Trans. Inf. Theory,} vol. 53, no. 10, pp. 3453-3475, Oct. 2007.
\bibitem {EliaKumar} P. Elia, K. Vinodh, M. Anand and P. V. Kumar, ``D-MG tradeoff and optimal codes for a class of AF and DF cooperative communication protocols,'' \textit{IEEE Trans. Inf. Theory}, to appear.
\bibitem {Boyd} S. Boyd and L. Vandenberghe, \textit{Convex Optimization}, New York: Cambridge University Press, 2004.
\end{thebibliography}
\end{document}